\documentclass[conference]{IEEEtran}
\IEEEoverridecommandlockouts

\usepackage[left=0.646in,right=0.583in,top=0.764in,bottom=1.0in]{geometry}

\usepackage{cite}
\usepackage{amsmath,amssymb,amsfonts}
\usepackage{algorithmic}
\usepackage{textcomp}
\usepackage{xcolor}
\usepackage{cite} 
\usepackage{url} 
\usepackage{amsfonts}
\usepackage{amssymb} 
\usepackage{mathrsfs}
\usepackage{euscript}
\usepackage{graphicx}
\usepackage{multirow}
\usepackage{algorithm}
\usepackage{algorithmic}
\usepackage{changepage}
\usepackage{color}
\usepackage{todonotes}
\usepackage{scalerel}
\usepackage{tikz}
\usetikzlibrary{svg.path}

\usepackage{balance}
\makeatletter
\patchcmd{\@makecaption}
{\scshape}
{}
{}
{}
\makeatletter
\patchcmd{\@makecaption}
{\\}
{.\ }
{}
{}
\makeatother

\definecolor{orcidlogocol}{HTML}{A6CE39}
\tikzset{
orcidlogo/.pic={
\fill[orcidlogocol] svg{M256,128c0,70.7-57.3,128-128,128C57.3,256,0,198.7,0,128C0,57.3,57.3,0,128,0C198.7,0,256,57.3,256,128z};
\fill[white] svg{M86.3,186.2H70.9V79.1h15.4v48.4V186.2z}
svg{M108.9,79.1h41.6c39.6,0,57,28.3,57,53.6c0,27.5-21.5,53.6-56.8,53.6h-41.8V79.1z M124.3,172.4h24.5c34.9,0,42.9-26.5,42.9-39.7c0-21.5-13.7-39.7-43.7-39.7h-23.7V172.4z}
svg{M88.7,56.8c0,5.5-4.5,10.1-10.1,10.1c-5.6,0-10.1-4.6-10.1-10.1c0-5.6,4.5-10.1,10.1-10.1C84.2,46.7,88.7,51.3,88.7,56.8z};
}
}
\newcommand\orcidicon[1]{\href{https://orcid.org/#1}{\mbox{\scalerel*{
\begin{tikzpicture}[yscale=-1,transform shape]
\pic{orcidlogo};
\end{tikzpicture}
}{|}}}}
\usepackage[colorlinks=true,linkcolor=blue,citecolor=blue]{hyperref}
\newtheorem{proof}{Proof}

\newtheorem{proposition}{Proposition}
\newtheorem{lemma}{Lemma}

\usepackage[most]{tcolorbox}

\begin{document} 
\author{
\IEEEauthorblockN{
~Jalal Jalali$^{\dag}$,
~Filip Lemic$^{\sharp}$,
~Hina Tabassum$^{\star}$,
~Rafael Berkvens$^{\dag}$,~and
Jeroen Famaey$^{\dag}$}
$^{\dag}$\text{ID}Lab, University of Antwerp - imec, Sint-Pietersvliet 7, 2000 Antwerp, Belgium \\
$^{\sharp}$AI-Driven Systems Lab, i2Cat Foundation, Barcelona, Spain\\
$^{\star}$Department of Electrical Engineering and Computer Science, York University, Toronto, ON, Canada\\
Email:  \href{mailto:jalal.jalali@uantwerpen.be}{\texttt{jalal.jalali@uantwerpen.be}}
\vspace{0mm} 
}

\title{\huge Toward Energy Efficient Multiuser IRS-Assisted URLLC Systems: A Novel Rank Relaxation Method\vspace{-0mm}}

\maketitle

\begin{abstract}
This paper proposes an energy efficient resource allocation design algorithm for an intelligent reflecting surface (IRS)-assisted downlink ultra-reliable low-latency communication (URLLC) network. This setup features a multi-antenna base station (BS) transmitting data traffic to a group of URLLC users with short packet lengths. We maximize the total network's energy efficiency (EE) through the optimization of active beamformers at the BS and passive beamformers (a.k.a phase shifts) at the IRS. The main non-convex problem is divided into two sub-problems. An alternating optimization (AO) approach is then used to solve the problem. Through the use of the successive convex approximation (SCA) with a novel iterative rank relaxation method, we construct a concave-convex objective function for each sub-problem. The first sub-problem is a fractional program that is solved using the Dinkelbach method and a penalty-based approach. The second sub-problem is then solved based on semi-definite programming (SDP) and the penalty-based approach. The iterative solution gradually approaches the rank-one for both the active beamforming and unit modulus IRS phase-shift sub-problems. Our results demonstrate the efficacy of the proposed solution compared to existing benchmarks.

\end{abstract}

\vspace{-0mm}
\section{Introduction}
\IEEEPARstart{I}{ntelligent reflecting surfaces} (IRSs) have recently emerged as a cost-effective and energy efficient technology to reconfigure wireless environments through passive reflections. Unlike conventional relays, IRSs do not require signal decoding or amplification~\cite{8910627}. The integration of IRSs into the communication network offers enhanced reliability, reduced packet re-transmissions, and latency. Consequently, IRSs can efficiently realize ultra-reliable low-latency communications \!(URLLC) with finite block length regime~\cite{8472907}. 
Unlike infinite blocklength regime, the achievable rate in the finite blocklength regime is a complex function of the signal-to-interference-plus-noise ratio \!(SINR), blocklength, and the probability of decoding errors~\cite{5452208}.

Recently, a handful of research studies investigated resource allocation in IRS-assisted URLLC networks. For instance, Hashemi \emph{et al.} presented the performance analysis of the average achievable data-rate and error probability over an IRS-aided URLLC transmission with and without phase noise~\cite{9517013}. Considering non-linear energy harvesting, the end-to-end performance of the IRS-assisted wireless system was analyzed in~\cite{9499071} for industrial URLLC applications, and the approximate closed-form expression of the block error rate was derived. 
Xie \emph{et al.} studied an IRS-assisted downlink multiuser URLLC system and jointly optimized the user grouping and the blocklength allocation at the base station 
(BS), as well as the reflective beamforming at the IRS for latency minimization~\cite{9519632}. 
Li \emph{et al.} investigated maximizing uplink data-rate in a URLLC system by jointly optimizing user transmit power, BS receive beamformer, and IRS reflection phase shifts under given error probability and blocklength conditions~\cite{10103618}. 
Very recently, an IRS-assisted machine-type communication network was considered in \cite{10100913} where multi-objective optimization was applied to simultaneously maximize EE and the number of admitted users with short packet length. The problem is solved by using fractional programming approach, i.e., quadratic transformation.
where a non-practical signal-to-interference-plus-noise ratio (SINR) bound was considered with no hard constraint on~the delay~\cite{10100913}
None of the aforementioned research works considered an \textit{energy efficient design of IRS-assisted URLLC} networks, although other performance metrics were considered in~\cite{9638478,10014778,9599656,sunn,10103618,10118968,VTC_accepted}. 
 
In this paper, we design an energy efficient beamforming and phase-shift optimization algorithm for downlink IRS-assisted URLLC services, wherein a multi-antenna BS serves multiple single-antenna URLLC receivers with the support of an IRS. To this end, we formulate an energy efficiency maximization problem, subject to the minimum required data-rate (in the finite blocklength regime) for each URLLC user and unit-modulus constraints at the IRS. We decompose the problem into two sub-problems and apply an alternating optimization (AO) approach to optimize both active and passive beamformers at the BS and IRS, respectively. Using an innovative iterative rank relaxation and successive convex approximation (SCA) method, we construct a concave-convex objective function for each sub-problem. Then, the first sub-problem (which is also a fractional program) is solved by applying the Dinkelbach method and a penalty-based approach; whereas, the second sub-problem is solved by applying the penalty-based approach and semi-definite programming (SDP).
The iterative solution gradually approaches the rank-one for both the active beamforming and unit modulus IRS phase-shift sub-problems. 
Our simulation results indicate the benefits of the IRS in enhancing the energy efficiency of URLLC systems and provide a comparison with benchmarks. 

\vspace{2.5mm}
\textit{Notations:} 
In this work, matrices and vectors are denoted by boldface capital letters, like $\mathbf{A}$, and lower-case letters, such as $\mathbf{a}$, respectively. 
If we consider a square matrix $\mathbf{A}$, $\mathbf{A}^T$, $\mathbf{A}^H$, ${\rm{rank}}(\mathbf{A})$, $\text{Tr}(\mathbf{A})$, and $||\mathbf{A}||_{*}$ denote its transpose, Hermitian conjugate transpose, rank, trace, and norm, respectively. 
$\mathbf{I}_N$ represents the $N \times N$ identity matrix. The diagonalization operation is represented by $\rm{diag}(\cdot)$. 
$\Re\{\cdot\}$ is used to denote the real part of a complex function. 
We express the absolute value of a complex scalar and the Euclidean norm of a complex vector by $|\cdot|$. 
The distribution of a circularly symmetric complex Gaussian (CSCG) random vector with mean $\boldsymbol{\mu}$ and covariance matrix $\mathbf{C}$ is denoted as $\mathcal{C}\mathcal{N}(\boldsymbol{\mu},\mathbf{C})$. 
$Q^{-1}(\cdot)$ represents the inverse of the Gaussian Q-function. 
$\mathbb{C}^{M\times N}$ denotes an $M\times N$ dimensional complex matrix, and $\nabla_{\mathbf{x}}$ expresses the gradient vector with respect to~$\mathbf{x}$.
Finally, we express ${N\times N}$ positive semi-definite matrices as $\mathbf{A} \in \mathbb{S}^{N}_{+}$ and read $ \mathbf{A}\succeq\mathbf{0}$.

This paper is structured as follows: In Section~\ref{section_2}, we introduce the system model and outline the proposed EE optimization problem. 
The resource allocation algorithm to solve the EE problem is detailed in Section~\ref{section_4}. 
In Section~\ref{section_5}, we assess the performance of our novel algorithm. 
Finally, Section~\ref{section_6} draws conclusions.

\begin{figure}
\centering
\includegraphics[width=0.95\linewidth]{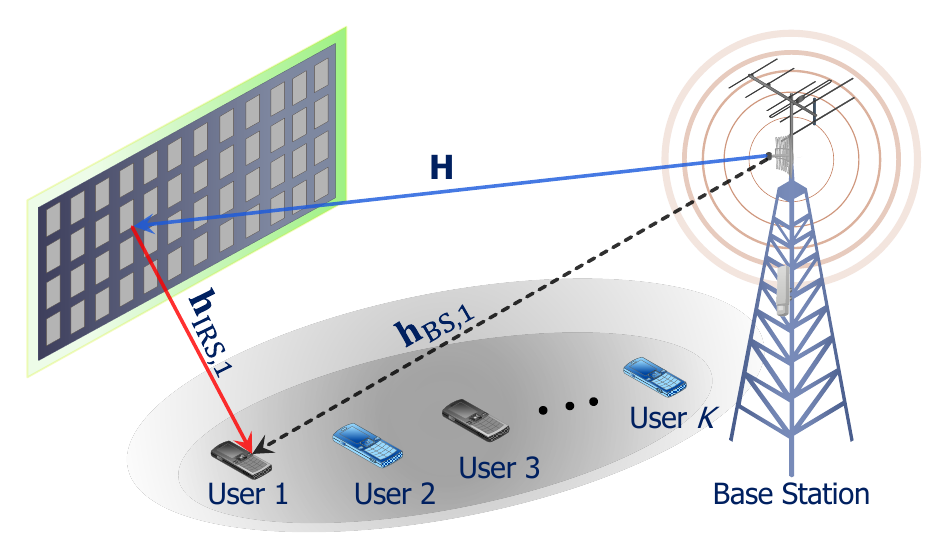}
\label{fig:system_model}
\caption{Illustration of a multi-user IRS-assisted downlink network comprising one BS and $K$ URLLC users employing finite block length transmissions. ${\mathbf{H}}$ is the channel matrix between the BS and the IRS, ${\mathbf{h}}_{\rm{IRS},1}$ denotes the channel response vector from the IRS to user 1, and $\mathbf{h}_{\rm{BS},1}$ indicates the channel between the BS and user 1.}
\end{figure}

\section{System Model and Problem Formulation} \label{section_2}
We consider a downlink multiuser URLLC-enabled IRS-assisted system, represented in Fig. \ref{fig:system_model}, which comprises an IRS with $N$ elements, a BS equipped with $M$ antennas, and $K$ users, each with a single antenna.
The group of IRS elements, BS antennas, and users are denoted by the sets $\mathcal{N}=\{1,...,N\}$, $\mathcal{M}=\{1,...,M\}$, and $\mathcal{K}=\{1,...,K\}$, respectively.
Additionally, we suppose that $B_k$ bits of information are assigned to user $k$. In this setup, the BS transforms these information bits into a block code in time slot $l$, characterized by a length of $m_d$ symbols. We decode the $k$-th user block code as $x_{k,l}$, with $l$ being part of the set $\mathcal{L} = \{1,2,...,m_d\}$.

The transmission signal emanating from the BS can be expressed as 
$\mathbf{s}{[l]}=
\sum_{k\in \mathcal{K}} 
{{\boldsymbol{\omega}_{k,l}}
{x_{k,l}}}$, 
where ${\boldsymbol{\omega}}_{k,l}\in \mathbb{C}^{M \times 1}$ 
represents the beamforming vector for user $k$.
The channel links exhibit time invariance (i.e., slow fading). Additionally, we assume that the BS possesses full access to the channel state information (CSI) and  URLLC users' delay requirements~\cite{10014778,10100913,Ghanem,10074428}. 
The baseband equivalent channel responses for BS-to-IRS, IRS-to-user $k$, and BS-to-user $k$ are denoted as
${\mathbf{H}}\in \mathbb{C}^{N \times M}$,
${\mathbf{h}}_{\rm{IRS},k}\in \mathbb{C}^{N\times 1}$, and 
$\mathbf{h}_{\rm{BS},k}\in \mathbb{C}^{M\times 1}$,
respectively. 
Also, we define the matrix of reflection coefficients at the IRS as
$\boldsymbol{\Psi}_l=
\text{diag} 
(\alpha_{1,l} e^{j\phi_{1,l}}, 
\alpha_{2,l} e^{j\phi_{2,l}},..., 
\alpha_{N,l} e^{j\phi_{N,l}})$, where 
$\alpha_{n,l} \in [0,1]$ and 
$\phi_{n,l} \in (0,2\pi]$, $\forall n \in \mathcal{N}, \forall l \in \mathcal{L}$ are the reflection amplitude and phase shift of the $n$-th reflection coefficient at the IRS during time slot $l$, respectively. 

By defining the combined channel link seen by the $k$-th URLLC user as $\mathbf{ h}_k^{H}\triangleq {\mathbf{h}}_{\rm{IRS},k}^{H}\boldsymbol{\Psi}_l {\mathbf{{\mathbf{H}}}}+\mathbf{ h}_{\rm{BS},k}^{H}$, 
$\forall k \in \mathcal{K},$ 
$\forall l \in \mathcal{L}$, the SINR of that user can be expressed as:
\begin{equation}\label{sinr}
\Gamma_{k,l}= 
\frac{{{\left| \mathbf{h}_{k}^{H} {\boldsymbol{\omega}}_{k,l} \right|}^2}}{\sum\limits_{i \ne k, i  \in \mathcal{K}}{{\left| \mathbf{h}_{k}^{H} {\boldsymbol{\omega}}_{i,l} \right|}^2} + \sigma _k^2}, 
\forall k \in \mathcal{K},
\forall l \in \mathcal{L},
\end{equation}
where $\sigma_{k}^2$ represents  the noise variance received at user $k$. 
In the context of URLLC systems, an accurate estimate of the achievable data-rate for each user can be delineated as follows~\cite{Blocklength}:
\begin{align}\label{rate}
\mathcal{R}_k(\boldsymbol{\omega}_{k,l},\boldsymbol{\Psi}_l)=
\mathcal{U}_k
(\boldsymbol{\omega}_{k,l},\boldsymbol{\Psi}_l)-
\mathcal{V}_k
(\boldsymbol{\omega}_{k,l},\boldsymbol{\Psi}_l),
~\forall k  \in  \mathcal{K},
\end{align}
where
\begin{align}
\mathcal{U}_k(\boldsymbol{\omega}_{k,l},\boldsymbol{\Psi}_l)
&=
\sum\limits_{l\in\mathcal{L}}\log_2 (1 + {\Gamma_{k,l} }),  
\ \ \ \ \ \ \ \ \forall k \in \mathcal{K},  
\end{align}
\begin{align}
\mathcal{V}_k(\boldsymbol{\omega}_{k,l},\boldsymbol{\Psi}_l)
&=
\sum\limits_{l\in\mathcal{L}}
\beta_k
\sqrt{\Delta_{k,l}},~~~~~~~~~~~~~~\forall k \in 
\mathcal{K},
\\
\beta_k
&=
\frac{Q^{-1}(\epsilon_k)}{\sqrt{m_d}},~~~~~~~~~~~~~~~~~~~\forall k \in 
\mathcal{K}.
\end{align}
Here, $\epsilon_k$ refers to the decoding error probability, $m_d$, as defined earlier, is representative of the blocklength, and $\Delta_{k,l}$ the channel dispersion, computed as:
\begin{equation}
\Delta_{k,l}=
(\log_2 e)^2
\left(
1-\frac{1}{\left(1+\Gamma_{k,l}\right)^2}
\right), 
\forall k \in \mathcal{K},
\forall l \in \mathcal{L}.   
\end{equation}
To ensure users' quality of service (QoS) regarding the received number of bits, the reliability, and the latency, a minimum data-rate denoted by $R_{\min,k}$ should be satisfied for each user as follows:
\begin{equation}
\mathcal{R}_k(\boldsymbol{\omega}_{k,l},\!\boldsymbol{\Psi}_l) \geq R_{\min,k},\: \forall k \in \mathcal{K}.
\end{equation}
Next, we describe the network's energy efficiency (EE) as the ratio of the total system data-rate over the associated network power consumption in~$\rm{[bits\slash Joule \slash Hz]}$:
\begin{equation}
\mathscr{\eta}_{\mathrm{eff}}(\boldsymbol{\omega}_{k,l},\! \boldsymbol{\Psi}_l)=
\frac{\sum\limits_{k\in \mathcal{K}}\mathcal{R}_k(\boldsymbol{\omega}_{k,l},\!\boldsymbol{\Psi}_l)}{\sum\limits_{k\in \mathcal{K}} {{\sum\limits_{l\in\mathcal{L}}{\left\| {\boldsymbol{\omega}}_{k,l} \right\|}^2}}
+
P_{\rm{IRS}}
+
NP_{d}
+
P_{\rm{c}}^{\rm{BS}}}, \label{EE_formula}
\end{equation}
where $P_{\rm{IRS}}$ represents the static power consumption necessary for sustaining the basic circuit operations of the IRS, while $P_d$ is the dynamic power expended per reflecting component, and $P_{\rm{c}}^{\rm{BS}}$ denotes the circuit power at the BS.

Our objective is to maximize the total EE by optimizing the active beamformers at the BS and the phase shifts at the IRS. Consequently, the problem can be articulated as follows:
\begin{subequations}
\begin{align}
\text{P}_1&: \mathop {\max}
\limits_{\boldsymbol{\omega}_{k,l},\boldsymbol{\Psi}_l} 
\: \:\: 
\mathscr{\eta}_{\mathrm{eff}}(\boldsymbol{\omega}_{k,l},\boldsymbol{\Psi}_l)
\nonumber\\
s.t.&:
\mathcal{R}_k(\boldsymbol{\omega}_{k,l},\!\boldsymbol{\Psi}_l) \geq R_{\min,k},
\forall k \in \mathcal{K},
\label{p3-3}\\
&~~
|\boldsymbol{\Psi}_{{nn},l}|=1,
~ ~\ \ \ \ \ \ \ \ ~~~~
\forall n \in \mathcal{N},
\forall l \in \mathcal{L},
\label{11g}\\
&~~
\sum
\limits_{l\in\mathcal{L}}
\sum\limits_{k\in \mathcal{K}} 
{\left\| 
{\boldsymbol{\omega}}_{k,l} 
\right\|}^2 
\leq 
p_{\max}, 
\label{p3-4}\\
&~~\boldsymbol{\omega}_{k,l} =0,
~ ~\ \ \ \ \ \ \ \ ~~~~~~~ \:
\forall k \in \mathcal{K}, 
\forall l \geq {\rm{T}}_{k}, 
\label{p4-4}\hspace{-2mm}
\end{align}
\end{subequations}
where \eqref{p3-3} is the minimum data-rate requirement ($R_{\min,k}$) of the $k$-th URLLC user.
Constraint \eqref{11g} seeks that the diagonal phase shift matrix contains $N$ unit-modulus elements along its main diagonal. The constraint \eqref{p3-4} defines the limitation of the BS transmission power budget, with $p_{\max}$ signifying the maximum permissible BS transmission power. Constraint \eqref{p4-4} is imposed to safeguard the real-time URLLC service functionality, ensuring that user $k$ receives service within the first ${\rm{T}}_{k}$ time slots to satisfy its delay requirements.

 Given the non-convex nature of the objective function and constraints in $\text{P}_1$, in the following, we propose an AO-based low-complexity solution.

\section{Energy Efficient Design for IRS-Assisted URLLC Network}\label{section_4}
In this section, we decompose $\text{P}_1$ into two sub-problems. The first sub-problem focuses on determining the optimal active beamformers at the BS, while the second sub-problem concentrates on optimizing the phase shifts at the IRS. 
We adopt the SCA technique to initially derive a lower bound for $\mathcal{R}_k(\boldsymbol{\omega}_{k,l},\boldsymbol{\Psi}_l)$, which will subsequently be utilized when addressing each sub-problem.


\begin{lemma}\label{lemma_1}
Let $(j)$ denote the subscript associated with the feasible solution procured in the $j$-th iteration of the SCA algorithm. A lower bound approximation of $\mathcal{R}_k(\boldsymbol{\omega}_{k,l},\boldsymbol{\Psi}_l)$ is given by:   
\begin{align}
&\mathcal{R}_k(\boldsymbol{\omega}_{k,l},\boldsymbol{\Psi}_l) 
\geq
\tilde{\mathcal{R}}_{k}(\boldsymbol{\omega}_{k,l},\boldsymbol{\Psi}_l)=
\sum\limits_{l\in\mathcal{L}}
\tilde{\mathcal{R}}_{k,l}^{(j)}(\boldsymbol{\omega}_{k,l},\boldsymbol{\Psi}_l),
\\
&\!\!\!
\tilde{\mathcal{R}}_{k,l}^{(j)}(\boldsymbol{\omega}_{k,l},\boldsymbol{\Psi}_l) 
\triangleq
\!-
\gamma_{k,l}\mathfrak{b}_{k,l}
\!+\!
2\rho_{k,l}
\Re\left\{\!\!
\sum\limits_{i \ne k, i  \in \mathcal{K}}
\mathfrak{a}_{i,k,l}^{(j)}\mathfrak{a}_{i,k,l}^H
\!+\!
\sigma^2_k
\!\right\}
\nonumber\\ 
&~~~~~~~~~~~~~~~~~~ \:
+2
\Re\left\{
\frac{\mathfrak{a}_{k,k,l}^{(j)}\mathfrak{a}_{k,k,l}^H}{{\mathfrak{b}_{k,l}^{(j)}(\mathfrak{b}_{k,l}^{(j)}-|\mathfrak{a}_{k,k,l}^{(j)}|^2)}}
\right\} + \xi_{k,l},
\nonumber\\ 
&~~~~~~~~~~~~~~~~~~~~~~~~~~~~~~~~~~~~~~~~~~~~~
\forall k \in \mathcal{K},
\forall l \in \mathcal{L},
\end{align}
where
\begin{align}
\mathfrak{a}_{k,i,l}
&= \mathbf{h}_{k}^{H} {\boldsymbol{\omega}}_{i,l},
\label{bound_c1}
\\
\mathfrak{b}_{k,l}
&=
{\sum\limits_{i  \in \mathcal{K}}{{\left| \mathbf{h}_{k}^{H} {\boldsymbol{\omega}}_{i,l} \right|}^2} + \sigma _k^2},
\label{bound_c2}
\\
\rho_{k,l}
&=
\frac{\beta_k(\log_2 e)^2}{\mathfrak{b}_{k,l}^{(j)}\sqrt{(\Gamma^{(j)}_{k,l})^2+2\Gamma^{(j)}_{k,l}}}
\label{bound_c3}
\\
\gamma_{k,l}
&=
\frac{|\mathfrak{a}_{k,k,l}^{(j)}|^2}
{\mathfrak{b}_{k,l}^{(j)}(\mathfrak{b}_{k,l}^{(j)}-|\mathfrak{a}_{k,k,l}^{(j)}|^2)}
+ \frac{\rho_{k,l}(\mathfrak{b}_{k,l}^{(j)}-|\mathfrak{a}_{k,k,l}^{(j)}|^2)}{\mathfrak{b}_{k,l}^{(j)}},
\label{bound_c4}
\end{align}

\begin{align} 
\xi_{k,l} 
&= \log_2 (1 + \Gamma_{k,l}^{(j)})-
\frac{|\mathfrak{a}_{k,k,l}^{(j)}|^2}
{\mathfrak{b}_{k,l}^{(j)}(\mathfrak{b}_{k,l}^{(j)}-|\mathfrak{a}_{k,k,l}^{(j)}|^2)}\nonumber\\
&-
\frac{\beta_k(2\Delta_{k,l}^{(j)}+(\log_2 e)^2)}{2\sqrt{\Delta_{k,l}^{(j)}}}-
\frac{\rho_{k,l}\mathfrak{b}_{k,l}^{(j)}}{2(1+\Gamma^{(j)}_{k,l})}.
\label{bound_c5}
\end{align}
In the $(j+1)$-th iteration, the data-rate function $\tilde{\mathcal{R}}_{k,l}(\boldsymbol{\omega}_{k,l},\boldsymbol{\Psi}_l)$ exhibits concavity with respect to each variable, reaching its boundary at the point $(\boldsymbol{\omega}_{k,l}^{(j)},\boldsymbol{\Psi}_l^{(j)})$.
\end{lemma}
\begin{proof}
See Appendix~\ref{Appen_A}.
\quad\quad\quad\quad\quad
\quad\quad\quad\quad\quad
\quad\quad\quad
\IEEEQEDhere
\end{proof}
The lower bound presented is more manageable compared to the original data-rate function $\tilde{\mathcal{R}}_k(\boldsymbol{\omega}_{k,l},\boldsymbol{\Psi}_l)$ in \eqref{rate}. 
However, this bound still involves coupled optimization variables. To address this, we implement the AO approach. Specifically, we optimize $\boldsymbol{\omega}_{k,l}$ and $ \boldsymbol{\Psi}_l$ by alternatively refining each variable while keeping the others constant.

\subsection{ Sub-problem~$^{(1)}$: Optimizing $\boldsymbol{\omega}_{k,l}$ with given $\boldsymbol{\Psi}_l$}
\vspace{1mm}
In this sub-problem, 
we presume that the passive reflecting elements at the IRS,
i.e., $\boldsymbol{\Psi}_l$
are fixed, and we proceed with designing the active beamformers, $\boldsymbol{\omega}_{k,l}$, at the BS.
By employing the principles of semi-definite programming (SDP), we derive the following:
$\boldsymbol{\Omega}_{k,l}= \boldsymbol{\omega}_{k,l}\boldsymbol{\omega}_{k,l}^H\in \mathbb{S}^{M}_{+}$ 
and 
$\mathbf{H}_k= \mathbf{h}_k\mathbf{h}_k^H \in \mathbb{S}^{M}_{+}$,
$\forall k \in \mathcal{K}$.
Leveraging on the aforementioned SDP relaxation, we can rewrite the SINR in \eqref{sinr} as follows:
\vspace{2mm}
\begin{equation}\label{sinrr}
\Gamma_{k,l}= \frac{{{ \text{Tr}({\mathbf{{\rm{ {\mathbf{H}}}}}}_{k}^{H} \boldsymbol{\Omega}_{k,l} )}}}{{\sum\limits_{i  \in \mathcal{K},i \ne k}  {{ \text{Tr}({\mathbf{{\rm{ {\mathbf{H}}}}}}_{k}^{H} \boldsymbol{\Omega}_{i,l}) } + \sigma _k^2} }}, 
\forall k \in \mathcal{K},
\forall l \in \mathcal{L},
\end{equation}
Consequently, the data-rate function defined in \eqref{rate} can be reshaped as $\mathcal{R}_k(\boldsymbol{\Omega}_{k,l})=\mathcal{U}_k(\boldsymbol{\Omega}_{k,l})-\mathcal{V}_k(\boldsymbol{\Omega}_{k,l})$. 
With the aid of the SDP transformations, the EE optimization problem $\text{P}_1$ can be reformulated as follows:
\vspace{1mm}
\begin{subequations}
\begin{align}
\text{P}_2&: \mathop {\max}\limits_{\boldsymbol{\Omega}_{k,l}} \: \:\: \mathscr{\eta}_{\mathrm{eff}}(\boldsymbol{\Omega}_{k,l})\nonumber\\
s.t.&:~
\mathcal{R}_k(\boldsymbol{\Omega}_{k,l})\geq R_{\min,k},
\forall k \in \mathcal{K},
\label{p2-1}\\
&~~~\boldsymbol{\Omega}_{k,l}\succeq\mathbf{0}, 
\ \ \ \ \ \ \ \ \ \ \ \ \: \:
\forall k \in \mathcal{K},
\forall l \in \mathcal{L},
\label{p2-2}\\
&~~~{\rm{rank}}(\boldsymbol{\Omega}_{k,l})\leq 1, \ \  \ \: \ \  
\forall k \in \mathcal{K},
\forall l \in \mathcal{L}, 
\label{p2-3}\\
&~~~\sum\limits_{l\in\mathcal{L}}\sum\limits_{k\in \mathcal{K}}\text{Tr}(\boldsymbol{\Omega}_{k,l}) \leq p_{{\max}}
\label{p2_4},\\
&~~~\text{Tr}(\boldsymbol{\Omega}_{k,l})= 0,~ \ \ \ \ \ \ \  ~\forall k \in \mathcal{K}, \forall l\geq {\rm{T}}_{k}
\label{p2_5}.
\end{align}
\end{subequations}
\vspace{1mm}
It is important to note that the constraint \eqref{p2-1} in $\text{P}_2$ does not exhibit concavity. To circumvent this non-concavity, we employ the outcome from \textit{Lemma \ref{lemma_1}}. A proposed surrogate lower bound of $\mathcal{R}_k(\boldsymbol{\Omega}_{k,l})$ is introduced, where it is guaranteed that $\mathcal{R}_k(\boldsymbol{\Omega}_{k,l})\geq \tilde{\mathcal{R}}_k(\boldsymbol{\Omega}_{k,l})$. 
Hence, we can re-express $\text{P}_2$ as:
\vspace{1mm}
\begin{subequations}
\begin{align}
\text{P}_3&: \mathop {\max}\limits_{\boldsymbol{\Omega}_{k,l}} \: \:\: 
\tilde{\mathscr{\eta}}_{\mathrm{eff}}(\boldsymbol{\Omega}_{k,l})
=
\frac{\sum\limits_{k\in \mathcal{K}}\tilde{\mathcal{R}}_k(\boldsymbol{\Omega}_{k,l})}{\mathcal{E}_{\rm{tot}}(\boldsymbol{\Omega}_{k,l})}
\nonumber\\
s.t.&:~
\tilde{\mathcal{R}}_k(\boldsymbol{\Omega}_{k,l})\geq R_{\min,k},  \ \ \ \ \ \forall k \in \mathcal{K},\label{p3-c1}\\
&~~~\eqref{p2-2}-\eqref{p2_5}. \nonumber
\end{align}
\end{subequations}
\vspace{1mm}
where
$\mathcal{E}_{\rm{tot}}(\boldsymbol{\Omega}_{k,l})=
\sum_{l\in\mathcal{L}}
\sum_{k\in \mathcal{K}}
\text{Tr}
(\boldsymbol{\Omega}_{k,l})+
P_{\rm{IRS}}+
NP_{d}+
P_{\text{c}}^{\text{BS}}$.
The numerator of the objective function and constraint \eqref{p3-c1} exhibit the form of convex-concave functions. Nevertheless,  the overall problem remains non-convex due to the presence of the non-convex rank-one constraint \eqref{p2-3} and the fractional objective function. 
Our forthcoming focus firstly centers on resolving the non-convex rank-one constraint \eqref{p2-3}, which emerges within the context of our optimization problem.
Fortunately, an innovative method has been proposed in~\cite{doi:10.1137/17M1147214}, specifically devised to address these types of rank constraints. 
We hereby direct our attention to the ensuing proposition. 
\begin{proposition}\label{propos_rank}
Consider a nonzero positive semi-definite  beamforming matrix, denoted as 
$\boldsymbol{\Omega}_{k,l} \in \mathbb{S}^{M}_{+}$, where $M$ represents the dimensionality of the matrix. 
We assert that $\boldsymbol{\Omega}_{k,l}$ is a rank one matrix if and only if the inequality 
$\varpi_{k,l}\mathbf{I}_{M-1} - \mho^{T}\boldsymbol{\Omega}_{k,l}\mho \succeq\mathbf{0}$ 
holds true, where $\varpi_{k,l} = 0$, 
$\mathbf{I}_{M-1}$ denotes an identity matrix with dimension $M-1$, and 
$\mho \in \mathbb{R}^{M \times (M-1)}$
corresponds to the eigenvectors corresponding to the $M-1$ smallest eigenvalues of $\boldsymbol{\Omega}_{k,l}$.
\end{proposition}
\begin{proof}
See Appendix~\ref{Appen_B}.
\quad\quad\quad\quad\quad
\quad\quad\quad\quad\quad
\quad\quad\quad
\IEEEQEDhere
\end{proof}
According to \textit{proposition \ref{propos_rank}}, we can replace the rank constraint \eqref{p3-c1} with a positive semi-definite constraint:
\begin{equation}\label{rank_one}
\varpi_{k,l}
\mathbf{I}_{M-1} - 
\mho^{{(q)}^T}
\boldsymbol{\Omega}_{k,l}
\mho^{(q)} 
\succeq\mathbf{0},
\forall k \in \mathcal{K},
\forall l \in \mathcal{L}.    
\end{equation} 
Given the unavailability of $\mho^{{(q)}^T}$, we resort to the SCA method and employ the smallest eigenvectors of $\boldsymbol{\Omega}_{k,l}$. 
To achieve the goal of ultimately having $\varpi_{k,l}=0$ while facilitating the attainment of an initial feasible point, we introduce a penalty term for $\varpi_{k,l}$ in the objective function. Thus, at SCA iteration $(q)$, the following convex problem is addressed: 
\begin{subequations}
\begin{align}
\text{P}_4&: 
\mathop {\max}
\limits_{\boldsymbol{\Omega}_{k,l},\varpi_{k,l}} 
\: \:\: 
\tilde{\mathscr{\eta}}_{\mathrm{eff}}(\boldsymbol{\Omega}_{k,l})
-
\varsigma^{(q)}
\sum_{l\in\mathcal{L}}
\sum_{k\in \mathcal{K}}
\varpi_{k,l}
\nonumber\\
s.t.&:~
\eqref{p2-2},\eqref{p2_4},
\eqref{p2_5},\eqref{p3-c1},
\eqref{rank_one}. 
\nonumber
\end{align}
\end{subequations}
Note that the penalty factor's updates follow an interior iterative sub-algorithm with the formula 
$\varsigma^{(q)}=
\min(\nu\varsigma^{(q-1)},\vartheta_{\max})$, 
where $\nu$ 
is a positive coefficient factor, and $\vartheta_{\max}$ is the maximum allowable penalty factor.

The objective function in $\text{P}_4$ is a fractional function, rendering the optimization problem to be still non-convex. 
To tackle this issue, we adopt Dinkelbach method~\cite{Dinkelbach}, with $(d)$ as its iteration, to address the fractional form of the optimization problem. 
Denoting 
$\tilde{\mathscr{\eta}}_{\mathrm{eff}}^*(\boldsymbol{\Omega}_{k,l})$
as the optimal energy efficiency point within the feasible solution set defined by its constraints, we then transform the problem into a non-fractional optimization problem for further resolution, i.e.,
\begin{subequations}
\begin{align}\hspace*{-3mm}
\text{P}_5&:
\!\!\!\!
\mathop {\max}
\limits_{\boldsymbol{\Omega}_{k,l},\varpi_{k,l}} 
\!\!
\sum\limits_{k\in \mathcal{K}} 
\!\tilde{\mathcal{R}}_{k}(\boldsymbol{\Omega}_{k,l}^{(d)})
\!-\!
\varrho^{(d)} 
\mathcal{E}_{\rm{tot}}(\boldsymbol{\Omega}_{k,l}^{(d)})
\!-\!
\varsigma^{(d)}
\!\sum_{l\in\mathcal{L}}
\!\sum_{k\in \mathcal{K}}
\!\varpi_{k,l}
\nonumber\\
\hspace{-10mm}
s.t.&:~
\eqref{p2-2},\eqref{p2_4},
\eqref{p2_5},\eqref{p3-c1},
\eqref{rank_one}, 
\nonumber
\end{align}
\end{subequations}
where $\varrho^{(d)}= \underset{\boldsymbol{\Omega}_{k,l}^{(d)}}{{\max}} 
~\tilde{\mathscr{\eta}}_{\mathrm{eff}}^{(d)}
(\boldsymbol{\Omega}_{k,l}^{(d)})$,
and $\varsigma^{(d)}$ is the new penalty factor update rules
$\varsigma^{(d)}=
\min(\nu\varsigma^{(d-1)},\vartheta_{\max})\times\mathcal{E}_{\rm{tot}}(\boldsymbol{\Omega}_{k,l}^{(d)})$. 
Finally, the convexity of the objective function in $\text{P}_5$ with respect to the active beamforming variables can be demonstrated through a formal proof based on the following proposition.

\begin{algorithm}[t]\vspace{5pt}
\caption{\strut Iterative SCA-based Resource Allocation Algorithm}
\begin{algorithmic}[1]\label{alg1}
\renewcommand{\algorithmicrequire}{\textbf{Input:}}
\renewcommand{\algorithmicensure}{\textbf{Output:}}
\REQUIRE 
Set ${d}=0$ for the Dinkelbach procedure, 
set the maximum number of iteration ${D}_{\max}$, 
initialize the beamformer matrix $\boldsymbol{\Omega}_{k,l}=\boldsymbol{\Omega}_{k,l}^{0}$,
initialize $\varsigma^{0}$,
set $\vartheta_{\max} \gg 1$, $\nu>1$, and set
the tolerance $\varepsilon=10^{-3}$.
\\    
\STATE \textbf{repeat}\\
\STATE \quad Calculate  $\tilde{\mathcal{R}}_{k}(\boldsymbol{\omega}_{k,l},\boldsymbol{\Psi}_l)$ for a given $\boldsymbol{\Psi}_l$.
\STATE \quad Calculate the rank-one relaxation constraint in \eqref{rank_one}.
\STATE \quad Solve $\text{P}_5$ for $ \varrho^{(d-1)} $.\\
\STATE \quad \quad \textbf{if}
$ \big|
\sum\limits_{k\in \mathcal{K}} 
\!\tilde{\mathcal{R}}_{k}(\boldsymbol{\Omega}_{k,l}^{(d)})
-
\varrho^{(d-1)} 
\mathcal{E}_{\rm{tot}}(\boldsymbol{\Omega}_{k,l}^{(d)})$\\
\quad ~\quad\quad\quad\quad\quad\quad\quad\quad$
-
\varsigma^{(d)}
\sum\limits_{l\in\mathcal{L}}
\sum\limits_{k\in \mathcal{K}}
\varpi_{k,l} 
\big|
\!\leq\! 
\varepsilon $\\
\STATE\quad \quad \quad \textbf{return} 
$ \boldsymbol{\Omega}_{k,l}= \boldsymbol{\Omega}_{k,l}^{(d)}$, 
$ \varrho^{*}=\varrho^{(d-1)} $.\\
\STATE\quad \quad \textbf{else}  $ \varrho^{(d)}=  \tilde{\mathscr{\eta}}_{eff}^{(d)}
(\boldsymbol{\Omega}_{k,l}^{(d)})$,
\textbf{end if}.\\
\STATE\quad Update $\varsigma^{(d)}=
\min(\nu\varsigma^{(d-1)},\vartheta_{\max})$.  \\
\STATE\quad ${d}\leftarrow {d}+1$.  \\
\STATE\quad \textbf{until} ${d}={D}_{{\max}}$.
\STATE   \textbf{return} $\boldsymbol{\Omega}_{k,l}=\boldsymbol{\Omega}_{k,l}^{*}$.
\end{algorithmic}\vspace*{3pt}
\end{algorithm}
\begin{proposition}\label{prob_optimal}
The optimal energy efficiency, denoted as $\tilde{\mathscr{\eta}}_{eff}^{*}
(\boldsymbol{\Omega}_{k,l}^{*})$, serves as a means to derive the resource allocation policy if and only if 
\begin{align}
\mathop {\max} 
\limits_{\boldsymbol{\Omega}_{k,l}} \: 
&
\sum\limits_{k\in \mathcal{K}} 
\tilde{\mathcal{R}}_{k}(\boldsymbol{\Omega}_{k,l}^{(d)})-
\varrho^{(d)} 
\mathcal{E}_{\rm{tot}}(\boldsymbol{\Omega}_{k,l}^{(d)}
-\!
\varsigma^{(q)}
\!\sum_{l\in\mathcal{L}}
\!\sum_{k\in \mathcal{K}}
\!\varpi_{k,l}=
\nonumber\\
&
\sum\limits_{k\in \mathcal{K}} 
\tilde{\mathcal{R}}_{k}(\boldsymbol{\Omega}_{k,l}^{*})-
\varrho^{*}~ 
\mathcal{E}_{\rm{tot}}(\boldsymbol{\Omega}_{k,l}^{*})
-
\varsigma^{(q)}
\!\sum_{l\in\mathcal{L}}
\!\sum_{k\in \mathcal{K}}
\!\varpi_{k,l}
=0,
\end{align}
for 
$\sum_{k\in \mathcal{K}} 
\tilde{\mathcal{R}}_{k}(\boldsymbol{\Omega}_{k,l}^{*})\geq 0$
and
$\mathcal{E}_{\rm{tot}}(\boldsymbol{\Omega}_{k,l}^{*})\geq0$,
where $\boldsymbol{\Omega}_{k,l}^{*}$ provides  the optimal solution to $\text{P}_5$. 
\end{proposition}
\begin{proof}
See Appendix~\ref{Appen_C}.
\quad\quad\quad\quad\quad
\quad\quad\quad\quad\quad
\quad\quad\quad
\IEEEQEDhere
\end{proof}
The optimization problem $\text{P}_5$ is now a convex problem and can be efficiently solved by standard convex optimization solvers such as CVX~\cite{citeulike:163662,archix,9123410}. 
Ultimately, we detail our proposed algorithm in~\textbf{Algorithm~\ref{alg1}}.

\subsection{ Sub-problem~$^{(2)}$: Optimizing $\boldsymbol{\Psi}_l$ with given $\boldsymbol{\omega}_{k,l}$}

In the second sub-problem, equipped with the optimal active beamforming matrices $\boldsymbol{\omega}_{k,l}$ from the preceding sub-problem, we progress with the optimization of the passive reflecting elements at the IRS, denoted by $\boldsymbol{\Psi}_l$. 
Given that the semi-definite matrix $\boldsymbol{\Omega}_{k,l}$ is provided, the optimization problem in $\text{P}_1$ shifts towards the maximization of the data-rate. 
The major hurdle in optimizing the phase shifts at the IRS arises due to the constraint \eqref{11g}.
Particularly, the constraint \eqref{11g} imposes a unit-modulus constraint, which presents a significant challenge in the quest to solve the problem.
As such, we start by defining
${\boldsymbol{e}}=( e^{j\phi_{1}},..., e^{j\phi_{N}})^H\in\mathbb{C}^{N\times1}$
and 
$\tilde{{\boldsymbol{e}}}=[{\boldsymbol{e}}^{T} \: \psi]^T\in\mathbb{C}^{(N+1)\times1}$, where $\psi\in\mathbb{C}$ is a dummy variable with $|\psi|=1$. 
To aid in crafting the solution, we further define
${\mathbf{E}}=\tilde{{\boldsymbol{e}}}\tilde{{\boldsymbol{e}}}^{H}\in\mathbb{C}^{(N+1)\times(N+1)}$.
Consequently, we derive the following:
\begin{align}
\left|
({\mathbf{h}}_{\rm{IRS},k}^{H}\boldsymbol{\Psi}_l{\mathbf{\mathbf{H}}}+\mathbf{ h}_{\rm{BS},k}^{H})\boldsymbol{\Omega}_{k,l}
\right|^2 
&\triangleq 
\text{Tr}(\mathbf{E}\mathbf{Z}_k\boldsymbol{\Omega}_{k,l}\mathbf{Z}_k^{H})
\nonumber\\
&=
\text{Tr}(\boldsymbol{\Omega}_{k,l}\mathbf{Y}_k),
\forall k \in \mathcal{K},
\end{align}
where
$\mathbf{Z}_k=[(\text{diag}({\mathbf{h}}_{\rm{IRS},k}^{H})\mathbf{H})^{T} \: \: \: \mathbf{ h}_{\rm{BS},k}^*]^{T}, \mathbf{Y}_k=\mathbf{Z}_k^{H}\mathbf{E}\mathbf{Z}_k.$
In a manner akin to $\text{P}_3$, we address the non-convex constraint \eqref{p3-3} and the objective function. To achieve this, we employ \textit{Lemma \ref{lemma_1}} and subsequently rewrite the data-rate function as 
$\tilde{\mathcal{R}}_k(\boldsymbol{\Psi}_l)= \mathcal{U}_k(\boldsymbol{\Psi}_l)-\mathcal{V}_k(\boldsymbol{\Psi}_l)$.
Now, we restate the optimization problem $\text{P}_1$ as follows:
\begin{subequations}
\begin{align}\hspace{-6mm}
\text{P}_6&:\mathop {\max}
\limits_{\mathbf{E},\boldsymbol{\Psi}_l} \: \:\:
\sum\limits_{k\in \mathcal{K}} \tilde{R}_{k}(\boldsymbol{\Psi}_l)\nonumber\\
s.t.&:
\tilde{\mathcal{R}}_k(\boldsymbol{\Psi}_l)\geq R_{\min,k}, \:
\forall k \in \mathcal{K}, 
\label{p6-1} \\		
&~~
{\rm{diag}}(\boldsymbol{\Psi}_l)=
\mathbf{1}_{N+1},
~\forall l \in \mathcal{L},
\label{p6-2}\\
&~~{\rm{rank}}(\mathbf{E})
\leq 1,
\label{p6-3} \\
&~~\boldsymbol{\Psi}_l\succeq 0, 
\ ~~~~~~~~~~ ~~~
\forall l \in \mathcal{L}
\label{p6-4},\\
&~~ 
\mathbf{E}\succeq \mathbf{0}.
\label{p6-5}
\end{align}
\end{subequations}

Following a similar approach as in $\text{P}_6$, and drawing on the insights from \textit{proposition \ref{propos_rank}}, we can substitute the rank constraint \eqref{p6-3} with a positive semi-definite constraint:
\begin{equation}\label{rank_one_2}
\zeta_{l}
\mathbf{I}_{N} - 
\Upsilon^{{(g)}^T}
\boldsymbol{\Psi}_l
\Upsilon^{(g)} 
\succeq\mathbf{0},
\forall l \in \mathcal{L},    
\end{equation} 
where $\Upsilon^{(g)}$ is an $N \times (N + 1)$ matrix at the $(g)$-th SCA iteration with its columns corresponding to the smallest $N$ eigenvectors of $\boldsymbol{\Psi}_l$. 
Furthermore, to sustain a rank of one for $\boldsymbol{\Psi}_l$, equation \eqref{rank_one_2} must hold true, assuming $\zeta_l = 0$.
Hence, by incorporating $\zeta_{l}$ as a penalty term into the objective function of $\text{P}_6$, we write the subsequent convex problem:
\vspace{1pt}
\begin{subequations}
\begin{align}
\text{P}_7&: \underset{{\mathbf{E}},\boldsymbol{\Psi}_l} {\text{max}} 
\sum\limits_{k\in \mathcal{K}} 
{\tilde{\mathcal{R}}_k(\boldsymbol{\Psi}_l)}
-\varkappa^{(g)} 
\sum\limits_{l\in \mathcal{L}} 
\zeta_{l}
\nonumber\\
s.t.&:~
\eqref{p6-1},\eqref{p6-2},
\eqref{p6-4},\eqref{p6-5},
\eqref{rank_one_2},
\nonumber
\end{align}
\end{subequations}
where $\varkappa^{(g)}$ is a sequences of increasing weights. 
The optimization problem $\text{P}_7$ now can be efficiently solved just as $\text{P}_5$~\cite{citeulike:163662,10059879}. The solution of these two sub-problems yields the suboptimal solution of $\text{P}_1$. 
\begin{proposition}\label{prob_optimal_complete}
The objective function of $\text{P}_1$ is ensured to be monotonically non-decreasing throughout the iterations of the proposed algorithm.
\end{proposition}
\begin{proof}
See Appendix~\ref{Appen_D}.
\quad\quad\quad\quad\quad\quad\quad\quad\quad\quad\quad\quad\quad$\blacksquare$
\end{proof}

Ultimately, the final iterative AO algorithm that solves sub-problems is presented in \textbf{Algorithm \ref{alg_2}}.

\begin{algorithm}[t]
\caption{\strut Proposed Iterative AO algorithm}
\begin{algorithmic}[1]\label{alg_2}
\renewcommand{\algorithmicrequire}{\textbf{Input:}}
\renewcommand{\algorithmicensure}{\textbf{Output:}}
\REQUIRE Set $s=0$, maximum number of iteration ${S}_\text{max}$, and initialize the beamformer matrix and phase shifts as $\boldsymbol{\Omega}_{k,l}=\boldsymbol{\Omega}_{k,l}^{0}$ and $\boldsymbol{\Psi}_l=\boldsymbol{\Psi}_l^0$, respectively.\\    
\STATE \textbf{repeat}\\
\STATE \quad Solve problem $\text{P}_5$ for given $\boldsymbol{\Psi}_l^s$, and obtain the\\ \quad optimal solution $\boldsymbol{\Omega}_{k,l}^s$ based on~\textbf{Algorithm~\ref{alg1}}.
\STATE \quad Solve problem $\text{P}_7$ for given $\boldsymbol{\Omega}_{k,l}^s$, and obtain the\\ \quad optimal solution $\boldsymbol{\Psi}_l^s$.
\STATE \quad $s \leftarrow s+1$.
\STATE   \textbf{until} ${s}={S}_{\text{max}}$
\STATE   \textbf{return} $\{\boldsymbol{\Omega}_{k,l}^{{*}},\boldsymbol{\Psi}_l^*\}=$ $\{\boldsymbol{\Omega}_{k,l}^{s},\boldsymbol{\Psi}_l^s\}.$
\end{algorithmic}
\end{algorithm}	

\vspace{1mm}
\section{Numerical Results and Discussions}
\label{section_5}
This section provides simulation results to investigate the efficiency of the proposed algorithm in IRS-supported downlink URLLC systems, utilizing finite blocklength codes.
A rectangular region with dimensions of $(100,100)$ meters is considered. 
The BS is positioned at $(0, 0)$ meters, the IRS is located at $(50,0)$ meters, and all users are scattered randomly within the boundaries of this rectangle region. 
The path loss model is expressed as $35.3+37.6\log_{10}(d_k)~\rm{[dBm]}$, where $d_k$ refers to the distance in kilometers between the BS and the $k$-th URLLC user~\cite{10100913}. 
The convergence tolerance for the proposed AO algorithm, which is based on rank relaxation, SCA, and the Dinkelback method, is set at $10^{-2}$. 
It is assumed that the thermal noise density stands at $-174~\rm{[dBm\slash Hz]}$. 
In addition, the maximum probability of decoding error for URLLC user $k$, represented as $\epsilon_{k}$, is set to be $10^{-7}$.
Furthermore, the simulation parameters are configured with $K=4$, $M=5$, and $R_{\min,k}= 1.6~\rm{[bits\slash Sec\slash Hz]}$ , in accordance with the methodology outlined in references~\cite{10014778} and \cite{10100913}.

\begin{figure}[t]
\centering
\includegraphics[width=0.51\textwidth] {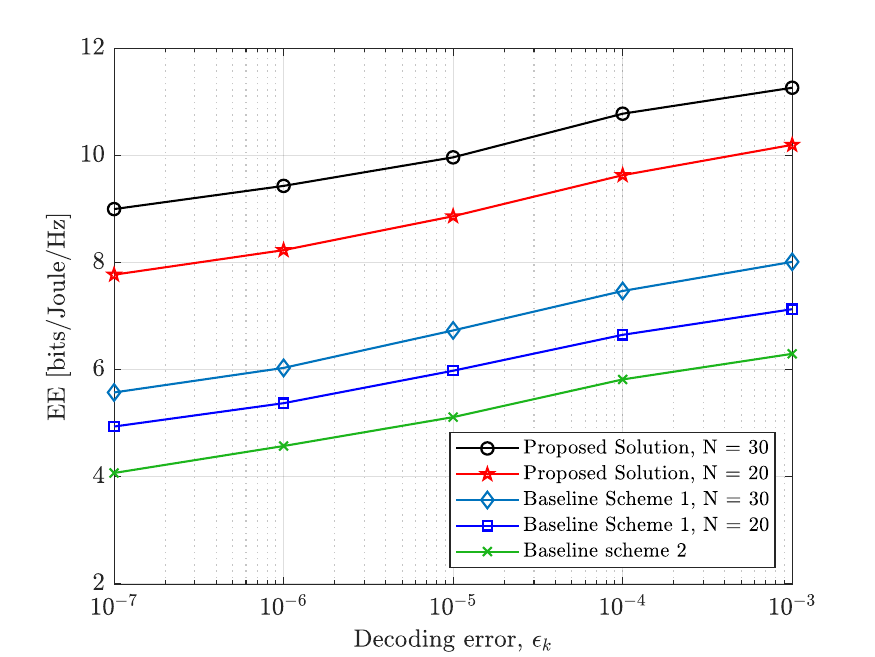}
\caption{\small Impact of tolerable decoding error,~$\epsilon_k$, on the EE.}\label{fig1}
\end{figure}

Fig. \ref{fig1} depicts the average EE in correlation with diverse maximum decoding error probability values, denoted as $\epsilon_{k,{\max}}$, keeping the block code length constant at $250$ symbols, represented as $m_d=250$. The plotted data brings to light an augmenting pattern in EE as the permissible decoding error broadens. This increment is attributable to the inverse proportionality of $Q^{-1}(\epsilon_k)$ and $\epsilon_k$. The dominating contributor to this rise is a diminished rate loss $\mathcal{V}_k(\boldsymbol{\omega}_{k,l},\boldsymbol{\Psi}_l)$ of \eqref{rate} correlating with an enlargement in the maximum acceptable decoding error probability. This correlation ensures that the minimal data-rate obligation is met at reduced transmission power, thereby amplifying EE. The figure also signifies an upsurge in EE when the number of reflecting elements is increased.

Moreover, the figure also includes a comparative study with two basic models. \textit{Baseline scheme 1} operates under a fixed beamforming policy at the IRS~\cite{10100913}, whereas \textit{Baseline scheme 2} considers a system without IRS. Our proposed methodology's performance, as demonstrated by the figure, surpasses that of baseline scheme 2 by $78.9\%$. This superior performance results from the incorporation of IRS and the harmonized optimization of active and passive beamforming matrices at the BS and IRS. This strategy displays an improved performance when compared to baseline scheme 1, $51.7\%$ and $48.8\%$ for $N=30$ and $N=20$, respectively, thereby underscoring the substantial benefits of integrating IRS into our system.

\begin{figure}[t]
\centering
\includegraphics[width=0.51\textwidth] {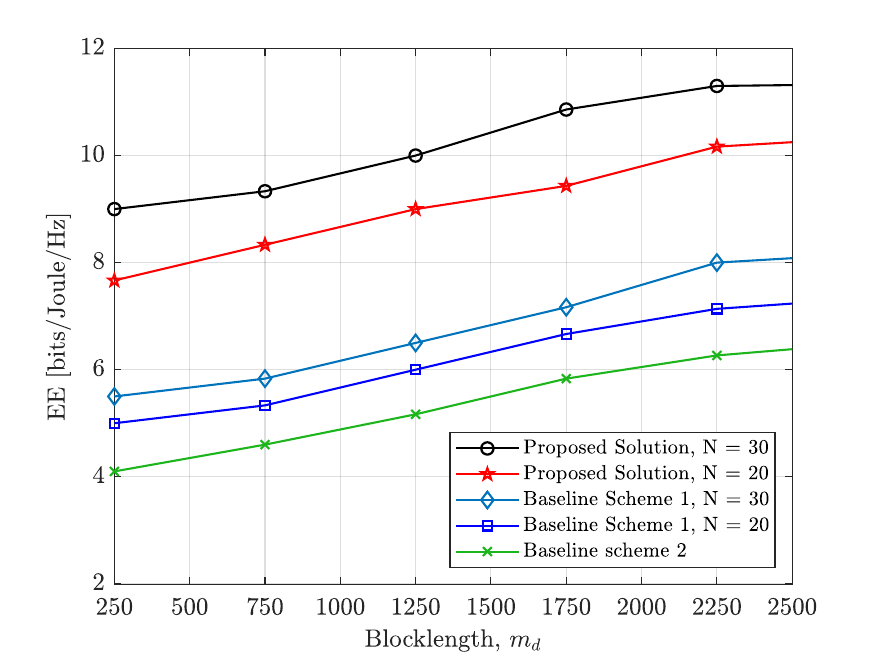}
\caption{\small Impact of the blocklength, $m_d$, on the EE.}\label{fig2}
\end{figure}  
Fig.~\ref{fig2} illustrates the impact of the blocklength size, denoted as $m_d$, on the achievable EE at a constant $\epsilon_{k}=10^{-7}$. 
With an increase in $m_d$, the EE initially experiences a moderate boost, followed by a slow elevation until it plateaus. In comparison to the scenario with an IRS, the scenario without it consistently exhibits lower EE, independent of the quantity of reflecting elements. 
Smaller $N$ values necessitate higher transmit power to meet QoS standards, leading to a reduction in EE. This reduction arises from the amplified interference in the data-rate function due to multi-user interference in the SINR function, resulting in an increase in transmit power, a decrease in data-rate, and consequently negatively pruned EE.

\section{Conclusion}\label{section_6}
This paper investigated the resource allocation for a downlink multiuser intelligent reflecting surface (IRS)-aided ultra-reliable low-latency communication (URLLC) system. In particular, the resource allocation design with active/passive beamforming was formulated to maximize the energy efficiency (EE) while considering URLLC Quality of Service (QoS) requirements for each user with adopting the short packet transmission. Subsequent to this, we brought forth an innovative iterative rank relaxation and adopted the successive convex approximation (SCA) approach, as well as a penalty-based method to solve each sub-problem. Simulation results demonstrated the efficacy of the proposed method compared to a system without or with unoptimized IRSs. Building upon these findings, our research could serve as a foundation for further studies on rank-constrained resource allocation strategies in comparable IRS-aided URLLC systems. 

\section{Acknowledgement}
This work was supported by the CHIST-ERA grant SAMBAS (CHIST-ERA20-SICT-003), with funding from FWO, ANR, NKFIH, and UKRI, and, in part, by the framework of the Recovery Plan, Transformation and Resilience (UNICO I+D 5G 2021, nr. TSI-063000-2021-6-Open6G Joint Open 6G Communications and Sensing), funded by the Spanish Ministry of Economic Affairs and Digital Transformation and European Union - NextGeneration EU.

\appendices
\section{Proof of the Lemma~\ref{lemma_1}}\label{Appen_A}
According to the data-rate equation in~\eqref{rate}, we can drive a lower bound on 
$\mathcal{U}_k
(\boldsymbol{\omega}_{k,l},\boldsymbol{\Psi}_l)$
and an upper bound on
$
\mathcal{V}_k
(\boldsymbol{\omega}_{k,l},\boldsymbol{\Psi}_l)$.
In order to approximate the non-convex logistic function, we utilize the Majorization Minimization (MM) algorithm \cite{Big-m,JalaliThesis,MM}. This is achieved by constructing a surrogate function using a first-order Taylor approximation, which is expressed as follows:
\begin{align}
\mathcal{U}_k(\textbf{a})
&\simeq
\mathcal{U}_k(\textbf{a}^{(j)})+
\nabla_{\textbf{a}}\mathcal{U}_k
(\textbf{a}^{(j)}).(\textbf{a}-\textbf{a}^{(j)})
\triangleq
\tilde{\mathcal{U}}_k(\textbf{a}), 
\label{F3-42-}
\\
\mathcal{V}_k(\textbf{b})
&\simeq
\mathcal{V}_k(\textbf{b}^{(j)})+
\nabla_{\textbf{b}}\mathcal{V}_k
(\textbf{b}^{(j)}).(\textbf{b}-\textbf{b}^{(j)})
\triangleq
\tilde{\mathcal{V}}_k(\textbf{b}),
\label{F3-42}
\end{align}
where $\textbf{a},\textbf{b}=\{\boldsymbol{\omega}_{k,l},\boldsymbol{\Psi}_l\}$ and $j$ is the iteration number.
Moreover, $\textbf{a}^{(j)}$ and $\textbf{b}^{(j)}$ denote the solutions of the problem at $(j)^{{th}}$ iteration.
Based on \eqref{bound_c1}-\eqref{bound_c5}, we first rewrite the SINR formula as follows: 
\begin{equation}\label{sinr_revise}
\Gamma_{k,l}= 
\frac{\left|\mathfrak{a}_{k,k,l}\right|^2}{\mathfrak{b}_{k,l}-\left|\mathfrak{a}_{k,k,l}\right|^2}, 
\forall k \in \mathcal{K},
\forall l \in \mathcal{L}.
\vspace*{2mm}
\end{equation}
Equivalently, we can rewrite the channel dispersion as:
\vspace{2mm}
\begin{align}
\Delta_{k,l}=
(\log_2 e)^2
\left(
1
-
\left(
\frac{\mathfrak{b}_{k,l}-\left|\mathfrak{a}_{k,k,l}\right|^2}{\mathfrak{b}_{k,l}}
\!
\right)^2
\right), 
\nonumber\\
\forall k \in \mathcal{K},
\forall l \in \mathcal{L}.   
\end{align}
Subsequently, we have: 
\begin{align}
\vspace*{1mm}
\mathcal{U}_k(\boldsymbol{\omega}_{k,l},\boldsymbol{\Psi}_l)
&=
-\sum\limits_{l\in\mathcal{L}}
\log_2 
\left(
1 -
\frac{\left|\mathfrak{a}_{k,k,l}\right|^2}{\mathfrak{b}_{k,l}}
\right)
\nonumber\\
&=
\sum\limits_{l\in\mathcal{L}}
\mathcal{U}_{k,l}(\boldsymbol{\omega}_{k,l},\boldsymbol{\Psi}_l),  
~~~~~~~~\forall k \in \mathcal{K},  
\end{align}
\begin{align}
\mathcal{V}_k(\boldsymbol{\omega}_{k,l},\boldsymbol{\Psi}_l)
&=
\sum\limits_{l\in\mathcal{L}}
\beta_k
\sqrt{\Delta_{k,l}}
\nonumber\\
&=
\sum\limits_{l\in\mathcal{L}}
\mathcal{V}_{k,l}(\boldsymbol{\omega}_{k,l},\boldsymbol{\Psi}_l),~~~~~~~~\forall k \in 
\mathcal{K}.
\end{align}
Thus, employing the MM technique allows us to derive the following lower limit for $\mathcal{U}_{k,l}(\boldsymbol{\omega}_{k,l},\boldsymbol{\Psi})$:
\begin{align}
\mathcal{U}_{k,l}(\boldsymbol{\omega}_{k,l},\boldsymbol{\Psi})
\geq
&-
\frac{\left|\mathfrak{a}_{k,k,l}^{(j)}\right|^2}
{\mathfrak{b}_{k,l}^{(j)}\left(\mathfrak{b}_{k,l}^{(j)}-\left|\mathfrak{a}_{k,k,l}^{(j)}\right|^2\right)}
\mathfrak{b}_{k,l}
\nonumber\\
&+
2\Re
\left\{
\frac{\mathfrak{a}_{k,k,l}^{(j)}\mathfrak{a}_{k,k,l}^{H}}
{\mathfrak{b}_{k,l}^{(j)}-\left|\mathfrak{a}_{k,k,l}^{(j)}\right|^2}
\right\}+ 
\xi_{k,l}^{\mathcal{U}},
\nonumber\\
&
~~~~~~~~~~
~~~~~~~~~~
\forall k \in \mathcal{K},
\forall l \in \mathcal{L},
\end{align}
where
\begin{equation}
\xi_{k,l}^{\mathcal{U}}=  
\log_2 (1 + \Gamma_{k,l}^{(j)})-
\frac{|\mathfrak{a}_{k,k,l}^{(j)}|^2}
{\mathfrak{b}_{k,l}^{(j)}(\mathfrak{b}_{k,l}^{(j)}-|\mathfrak{a}_{k,k,l}^{(j)}|^2)}.
\end{equation}
Similarly, we use the first-order Taylor approximation to derive the following upper limit for $\mathcal{V}_{k,l}(\boldsymbol{\omega}_{k,l},\boldsymbol{\Psi})$:
\begin{align}
\mathcal{V}_{k,l}(\boldsymbol{\omega}_{k,l},\boldsymbol{\Psi})
\leq
\beta_k
\sqrt{\Delta_{k,l}^{j}}
+
\frac{(\log_2 e)^2\beta_k}{2\sqrt{\Delta_{k,l}^{j}}}
&\left(\!
1-\frac{1}{\left(1+\Gamma_{k,l}\right)^2}
\!\right), 
\nonumber
\\
&\!\!\!\!\!\!\!\!
\forall k \in \mathcal{K},
\forall l \in \mathcal{L}.
\end{align}
Now, by defining $\Xi_{k,l}\triangleq  1\slash (1+\Gamma_{k,l})$ and according to~\cite{10100913}, we can rewrite: 
\begin{align}
\frac{1}{\left(1+\Gamma_{k,l}\right)^2}
&=
\Xi_{k,l}^2
\geq
2\Xi_{k,l}^{(j)}\Xi_{k,l}
-(\Xi_{k,l}^{(j)})^2
\nonumber\\
&=
\frac{2}{\left(1+\Gamma_{k,l}^{(j)}\right)\left(1\!+\!\Gamma_{k,l}\right)}
-
\frac{1}{\left(1\!+\!\Gamma_{k,l}^{(j)}\right)^2},
\nonumber\\
&~~~~~~~~~~~~~~
~~~~~~~~~~~~~~~
\forall k \in \mathcal{K},
\forall l \in \mathcal{L},
\label{MM_delta_2}
\end{align}
where the inequality holds due to convexity of $(\Xi_{k,l})^2$.
The term $1/(1+\Gamma_{k,l})$ in \eqref{MM_delta_2} is non-convex due to its fractional form. We rewrite this term as follows and then apply MM: 
\begin{align}
\frac{1}{1+\Gamma_{k,l}} = 
\frac{\mathfrak{b}_{k,l}-\left|\mathfrak{a}_{k,k,l}\right|^2}{\mathfrak{b}_{k,l}}
\geq
-\frac{\left(\mathfrak{b}_{k,l}^{(j)}-\big|\mathfrak{a}_{k,k,l}^{(j)}\big|^2\right)\mathfrak{b}_{k,l}}{\big|\mathfrak{b}_{k,l}^{(j)}\big|^2}
\nonumber\\
~~~~~+2\Re
\left\{
\frac{
\sum\limits_{i \ne k, i  \in \mathcal{K}}
\mathfrak{a}_{i,k,l}^{(j)}\mathfrak{a}_{i,k,l}^H
+
\sigma^2_k
}{\mathfrak{b}_{k,l}^{(j)}}
\right\}
,
\forall k \in \mathcal{K},
\forall l \in \mathcal{L},
\label{MM_delta_3}
\end{align}
where, again, the inequality holds due to the convexity with respect to the $\mathfrak{a}_{i,k,l}$ and $\mathfrak{b}_{k,l}$. Thus, we derive an upper bound  as:  
\begin{align}
\begin{split}
\mathcal{V}_{k,l}(\boldsymbol{\omega}_{k,l},\boldsymbol{\Psi})
&\leq
\frac{(\log_2 e)^2\beta_k}{\sqrt{(\Gamma^{(j)}_{k,l})^2+2\Gamma^{(j)}_{k,l}}}
\left(
\frac{\left(\mathfrak{b}_{k,l}^{(j)}-\big|\mathfrak{a}_{k,k,l}^{(j)}\big|^2\right)\mathfrak{b}_{k,l}}{\big|\mathfrak{b}_{k,l}^{(j)}\big|^2}
\nonumber
\right.
\\
&
\left.
~-2\Re
\left\{
\frac{
\sum\limits_{i \ne k, i  \in \mathcal{K}}
\mathfrak{a}_{i,k,l}^{(j)}\mathfrak{a}_{i,k,l}^H
+
\sigma^2_k
}{\mathfrak{b}_{k,l}^{(j)}}
\right\}
\right)
+
\xi_{k,l}^{\mathcal{V}},
\nonumber\\
&
~~~~~~~~~~~~~~~~~
~~~~~~~~~~~~~~~~~
\forall k \in \mathcal{K},
\forall l \in \mathcal{L},
\end{split}
\tag{39}
\end{align}
where
\begin{align}
\xi_{k,l}^{\mathcal{V}}
&= 
\frac{(2\Delta_{k,l}^{(j)}+(\log_2 e)^2)\beta_k}{2\sqrt{\Delta_{k,l}^{(j)}}}
\nonumber\\
&+
\frac{(\log_2 e)^2\beta_k}{2(1+\Gamma^{(j)}_{k,l})\sqrt{(\Gamma^{(j)}_{k,l})^2+2\Gamma^{(j)}_{k,l}}},
\forall k \in \mathcal{K},
\forall l \in \mathcal{L}.
\tag{40}
\end{align}
Ultimately, the lower bound of the data-rate function in~\eqref{rate} can be approximated as:
\begin{equation}
\tilde{\mathcal{R}}_{k}(\boldsymbol{\omega}_{k,l},\boldsymbol{\Psi}_l)
\triangleq 
\sum\limits_{l\in\mathcal{L}}
\big(
\mathcal{U}_{k,l}(\boldsymbol{\omega}_{k,l},\boldsymbol{\Psi})
-
\mathcal{V}_{k,l}(\boldsymbol{\omega}_{k,l},\boldsymbol{\Psi})
\big),
\forall k \in \mathcal{K}.
\tag{41}
\end{equation}
This completes the proof. 
\quad\quad\quad\quad\quad
\quad\quad\quad\quad\quad
\quad\quad\quad\quad
\IEEEQEDhere

\section{Proof of the Proposition~\ref{propos_rank}}\label{Appen_B}
Assuming the nonnegative eigenvalues of 
$\boldsymbol{\Omega}_{k,l}$
are arranged in descending order as 
$[\kappa_{M}; \kappa_{M-1}; ... ; \kappa_{1}]$, 
we can exploit the relationship between an eigenvector's Rayleigh quotient and its associated eigenvalue. Consequently, the matrix 
$\varpi_{k,l}\mathbf{I}_{M-1} - \mho^{T}\boldsymbol{\Omega}_{k,l}\mho,$
$\forall k \in \mathcal{K},
\forall l \in \mathcal{L},$
takes on the form of a diagonal matrix, with its diagonal elements set as 
$[\varpi_{k,l}-\kappa_{M-1}; \varpi_{k,l}-\kappa_{M-2}; ... ; \varpi_{k,l}-\kappa_{1}]$.
In light of this, we observe that 
$\boldsymbol{\Omega}_{k,l}$
possesses 
$M-1$ 
smallest eigenvalues all being zero if and only if the conditions 
$\varpi_{k,l}\mathbf{I}_{M-1} - \mho^{T}\boldsymbol{\Omega}_{k,l}\mho \succeq\mathbf{0}$
and 
$\varpi_{k,l} = 0,$
$\forall k \in \mathcal{K},
\forall l \in \mathcal{L},$
are satisfied simultaneously. 
Consequently, 
$\boldsymbol{\Omega}_{k,l}$
qualifies as a rank one matrix under these conditions~\cite{doi:10.1137/17M1147214}.~\quad\quad\quad\quad\quad\quad\quad\quad\quad\quad\quad\quad\quad\quad\quad\quad\quad\quad\quad\quad~\IEEEQEDhere

\section{Proof of the Proposition~\ref{prob_optimal}}\label{Appen_C}

Let's denote $\varrho^{*}$ be the optimal solution, corresponding to the optimal resource allocation policy $\boldsymbol{\Omega}^{*}_{k,l}$ of the objective function in $\text{P}_5$, 
that is:
\begin{align}
\varrho^{*} =
\mathop 
{\max}
\limits_{\boldsymbol{\Omega}_{k,l}} 
\: 
\frac{\sum\limits_{k\in \mathcal{K}}\tilde{\mathcal{R}}_k(\boldsymbol{\Omega}_{k,l})}{\mathcal{E}_{\rm{tot}}(\boldsymbol{\Omega}_{k,l})}
-
\varsigma^{(q)}
\sum_{l\in\mathcal{L}}
\sum_{k\in \mathcal{K}}
\varpi_{k,l}.
\tag{42}
\end{align}
Therefore, the optimal EE must satisfy the following inequality:
\begin{align}
\varrho^{*}=
&
\frac{
\sum\limits_{k\in \mathcal{K}}
\tilde{\mathcal{R}}^{*}_k(\boldsymbol{\Omega}^{*}_{k,l})}
{\mathcal{E}^{*}_{\rm{tot}}(\boldsymbol{\Omega}^{*}_{k,l})}
-
\varsigma^{(q)}
\sum_{l\in\mathcal{L}}
\sum_{k\in \mathcal{K}}
\varpi_{k,l}
\nonumber
\\
\geq &
\frac{
\sum\limits_{k\in \mathcal{K}}
\tilde{\mathcal{R}}^{}_k(\boldsymbol{\Omega}^{}_{k,l})}
{\mathcal{E}^{(d)}_{\rm{tot}}(\boldsymbol{\Omega}^{}_{k,l})}
-
\varsigma^{(q)}
\sum_{l\in\mathcal{L}}
\sum_{k\in \mathcal{K}}
\varpi_{k,l}.
\tag{43}
\end{align}
It becomes readily apparent that:
\begin{align}
& 
\sum\limits_{k\in \mathcal{K}} 
\!\!\tilde{\mathcal{R}}_{k}(\boldsymbol{\Omega}_{k,l}^{})
\!-\!\!
\varrho^{*}\!
\mathcal{E}_{\rm{tot}}(\boldsymbol{\Omega}_{k,l}^{})
\!-\!
\varsigma^{(q)}
\!\!\sum_{l\in\mathcal{L}}
\!\sum_{k\in \mathcal{K}}
\!\varpi_{k,l}
\mathcal{E}_{\rm{tot}}(\boldsymbol{\Omega}_{k,l}^{})
\!\leq 
\!0,
\tag{44}
\end{align}
Thus, it can be inferred that:
\begin{align}
& 
\sum\limits_{k\in \mathcal{K}} 
\!\!\tilde{\mathcal{R}}_{k}(\boldsymbol{\Omega}_{k,l}^{*})
\!-\!
\varrho^{*}\!
\mathcal{E}_{\rm{tot}}(\boldsymbol{\Omega}_{k,l}^{*})
\!-\!\!
\varsigma^{(q)}
\!\sum_{l\in\mathcal{L}}
\!\sum_{k\in \mathcal{K}}
\!\varpi_{k,l}
\mathcal{E}_{\rm{tot}}(\boldsymbol{\Omega}_{k,l}^{*})
\!= \!0,\!
\tag{45}
\end{align}
Therefore, we have:
\begin{align}
\mathop {\max} 
\limits_{\boldsymbol{\Omega}_{k,l}} \: 
& \sum\limits_{k\in \mathcal{K}} 
\tilde{\mathcal{R}}_{k}(\boldsymbol{\Omega}_{k,l}^{})-
\varrho^{*}~ 
\mathcal{E}_{\rm{tot}}(\boldsymbol{\Omega}_{k,l})
\nonumber\\
& ~~~~~~~~~~~~~~~~\:
-
\varsigma^{(q)}
\!\sum_{l\in\mathcal{L}}
\!\sum_{k\in \mathcal{K}}
\!\varpi_{k,l}
\mathcal{E}_{\rm{tot}}(\boldsymbol{\Omega}_{k,l}^{})
=0,
\tag{46}
\end{align}
and this can be attained through the resource allocation policy. This concludes the proof.
\quad\quad\quad\quad\quad
\quad\quad\quad\quad\quad 
\quad\quad\quad\quad \
$\blacksquare$

\section{Proof of the Proposition~\ref{prob_optimal_complete}}\label{Appen_D}
Let's denote the objective functions of $\text{P}_1$, $\text{P}_5$, and $\text{P}_7$ respectively as $\mathcal\daleth_{\text{P}_1}$, $\mathcal\daleth_{\text{P}_5}$, and $\mathcal\daleth_{\text{P}_7}$.
Also,
in consideration of $\{\boldsymbol{\Omega}_{k,l}^{s}, \boldsymbol{\Psi}_l^{s}\}$ and $\{\boldsymbol{\Omega}_{k,l}^{s-1}, \boldsymbol{\Psi}_l^{s-1}\}$ as the feasible solutions of $\text{P}_1$ in the $s$-th and $(s-1)$-th iterations respectively, we can establish the following inequalities:  
\begin{align}\hspace*{6mm}
\mathcal\daleth_{\text{P}_1}(\boldsymbol{\Omega}_{k,l}^{s},\boldsymbol{\Psi}_l^{s})
&= \mathcal\daleth_{\text{P}_7}(\boldsymbol{\Omega}_{k,l}^{*},\boldsymbol{\Psi}_l^{s})
\nonumber\\
&\geq
\mathcal\daleth_{\text{P}_7}( \boldsymbol{\Psi}_l^{s-1})
\nonumber\\
&=
\mathcal\daleth_{\text{P}_1}( \boldsymbol{\Omega}_{k,l}^{s},\boldsymbol{\Psi}_l^{s-1}), \label{c1}
\tag{47}
\end{align}
\begin{align}
\mathcal\daleth_{\text{P}_1}(\boldsymbol{\Omega}_{k,l}^{s},\boldsymbol{\Psi}_l^{s-1})
&=
\mathcal\daleth_{\text{P}_5}( \boldsymbol{\Omega}_{k,l}^{s},\boldsymbol{\Psi}_l^{*})
\nonumber \\
&\geq
\mathcal\daleth_{\text{P}_5}(\boldsymbol{\Omega}_{k,l}^{s-1})
\nonumber\\
&=
\mathcal\daleth_{\text{P}_1}( \boldsymbol{\Omega}_{k,l}^{s-1},\boldsymbol{\Psi}_l^{s}). 
\tag{48}
\label{c2}
\end{align}
By utilizing inequalities \eqref{c1} and \eqref{c2}, we can guarantee the improvement in the value of the objective function of $\text{P}_1$ after every iteration.
\quad\quad\quad\quad\quad\quad\quad\quad \quad\quad\quad\quad\quad\quad\quad\quad\quad\quad \ \ $\blacksquare$

\bibliography{ref}
\bibliographystyle{ieeetr}
\end{document}